\documentclass[a4paper,10pt]{article}

\usepackage[top=1.0in,right=1.in,left=1.in,bottom=1.75in]{geometry}
\usepackage{amssymb,amsbsy}
\usepackage{amsmath}
\usepackage{amsthm}
\usepackage{amsfonts}
\usepackage{graphics,graphicx}
\usepackage[T1]{fontenc}
\usepackage{color}
\usepackage{booktabs}
\usepackage[affil-it,auth-sc]{authblk}

\newcommand{\beq}{\begin{equation}}
\newcommand{\eeq}{\end{equation}}
\newcommand{\lb}{\label}
\newcommand{\beqar}{\begin{eqnarray}}
\newcommand{\eeqar}{\end{eqnarray}}
\newcommand{\bit}{\begin{itemize}}
\newcommand{\eit}{\end{itemize}}
\newcommand{\benum}{\begin{enumerate}}
\newcommand{\eenum}{\end{enumerate}}
\newcommand{\barr}{\begin{array}}
\newcommand{\earr}{\end{array}}
\def\ds{\displaystyle}

\newcommand{\ra}[1]{\renewcommand{\arraystretch}{#1}}

\newtheorem{theorem}{Theorem}[section]

\newcommand{\jump}[2]{[\mbox{\hspace{-#1em}}[#2]\mbox{\hspace{-#1em}}]}

\newcommand{\modI}{\text{I}}
\newcommand{\modII}{\text{II}}
\newcommand{\modIII}{\text{III}}

\def\XXint#1#2#3{{\setbox0=\hbox{$#1{#2#3}{\int}$}
   \vcenter{\hbox{$#2#3$}}\kern-.5\wd0}}

\def\b0{\mbox{\boldmath $0$}}

\def\be{\mbox{\boldmath $e$}}

\def\bn{\mbox{\boldmath $n$}}

\def\bp{\mbox{\boldmath $p$}}

\def\bu{\mbox{\boldmath $u$}}

\def\bx{\mbox{\boldmath $x$}}

\def\bA{\mbox{\boldmath $A$}}
\def\bB{\mbox{\boldmath $B$}}

\def\bE{\mbox{\boldmath $E$}}
\def\bF{\mbox{\boldmath $F$}}
\def\bG{\mbox{\boldmath $G$}}

\def\bI{\mbox{\boldmath $I$}}

\def\bR{\mbox{\boldmath $R$}}

\def\bU{\mbox{\boldmath $U$}}

\newcommand{\bsigma}{\mbox{\boldmath $\sigma$}}
\newcommand{\bSigma}{\mbox{\boldmath $\Sigma$}}

\newcommand{\bxi}{\mbox{\boldmath $\xi$}}

\def\f0{\ensuremath{\mathbb{O}}}

\newcommand{\mF}{\ensuremath{\mathcal{F}}}

\newcommand{\mI}{\ensuremath{\mathcal{I}}}

\newcommand{\mK}{\ensuremath{\mathcal{K}}}

\newcommand{\mP}{\ensuremath{\mathcal{P}}}
\newcommand{\mQ}{\ensuremath{\mathcal{Q}}}

\newcommand{\mS}{\ensuremath{\mathcal{S}}}

\newcommand{\bmA}{\mbox{\boldmath $\mathcal{A}$}}
\newcommand{\bmB}{\mbox{\boldmath $\mathcal{B}$}}

\newcommand{\bmE}{\mbox{\boldmath $\mathcal{E}$}}

\newcommand{\bmR}{\mbox{\boldmath $\mathcal{R}$}}

\def\Re{\mathop{\mathrm{Re}}}

\newcommand{\sign}{\mathop{\mathrm{sign}}}

\newcommand{\arctanh}{\mathop{\mathrm{arctanh}}}

\newcommand{\Reals}{\ensuremath{\mathbb{R}}}

\def\IJSS{{\it Int.\ J.\ Solids Struct.}\ }

\def\JAM{{\it ASME J.\ Appl.\ Mech.}\ }

\def\JMPS{{\it J.\ Mech.\ Phys.\ Solids}\ }

\def\PRSA{{\it Proc.\ R.\ Soc. A}\ }

\def\ZAMM{{\it Z.\ Angew.\ Math.\ Mech.}\ }

\title{Integral identities for a semi-infinite interfacial crack\\ in 2D and 3D elasticity}

\author{A. Piccolroaz\footnote{Corresponding author: e-mail: roaz@ing.unitn.it; phone: +39\,0461\,282583.}}
\author{G. Mishuris}

\affil{Institute of Mathematical and Physical Sciences, Aberystwyth University, Wales, U.K.}

\begin{document}

\maketitle

\begin{abstract}

The paper is concerned with the problem of a semi-infinite crack at the interface between two dissimilar elastic half-spaces, loaded by a general asymmetrical
system of forces distributed along the crack faces. On the basis of the weight function approach and the fundamental reciprocal identity (Betti formula), we
formulate the elasticity problem in terms of singular integral equations relating the applied loading and the resulting crack opening.
Such formulation is fundamental in the theory of elasticity and extensively used to solve several problems in linear elastic fracture mechanics (for instance
various classic crack problems in homogeneous and heterogeneous media). This formulation is also crucial in important recent multiphysics applications, where the
elastic problem is coupled with other concurrent physical phenomena. A paradigmatic example is hydraulic fracturing, where the elasticity equations are coupled with
fluid dynamics.

\vspace{3mm}
{\it Keywords:} Interfacial crack; Singular integral; Betti identity; Weight function; Hydraulic fracture.

\end{abstract}



\section{Introduction}
\label{sec1}

Three-dimensional problems of linear elasticity can be solved nowadays by a variety of techniques. Purely numerical methods, such as finite element and boundary
element techniques, are very popular and extensively used, thanks to the computational power of modern computers. However, analytical methods, such as integral
transforms, variational inequalities, and singular integral equations, are still of great interest, for instance in deriving fundamental solutions for some basic
geometries (to be employed in numerical methods), in the assertion of existence and uniqueness of the solution, and in bifurcation theory (Bigoni and Capuani, 2002;
Bigoni and Dal Corso, 2008).

In particular, the method of singular integral equations was first developed for two-dimensional problems (Muskhelishvili, 1953) and later extended to
three-dimensional problems by the rapidly developing theory of multi-dimensional singular integral operators (Kupradze et al., 1979; Mikhlin and Prossd\"{o}rf, 1980).
For crack problems, the general approach to derive the singular integral formulation is based on the Green's function method (Weaver, 1977; Budiansky and Rice, 1979,
Linkov et al., 1997). This method consists in two steps. First, the stresses at a point $\bx$ inside a body containing a discontinuity surface $S$ are expressed in
the form
\beq
\sigma_{ij} (\bx) = \int_S K_{ijk}(\bx,\bxi) \Delta u_k(\bxi) d\bxi,
\eeq
where $K_{ijk}(\bx,\bxi) = c_{ijst} (\partial/\partial x_t) T^+_{sk}(\bx,\bxi)$ and $T^+_{sk}(\bx,\bxi)$ is the traction associated with the Green's function
$U_{sp}(\bx,\bxi)$, according to the equation $T^+_{sk}(\bx,\bxi) = n^+_l c_{klpq} U_{sp,q}(\bx,\bxi)$. Here $c_{ijst}$ is the elasticity tensor, $n_l^+$ is the 
outward normal on the $+$ side of the surface and the displacement discontinuity is denoted by $\Delta u_k = u_k^- - u_k^+$. The second step consists in taking the 
limit as the point $\bx$ approaches a point $\bx_0$ on the surface $S$, to obtain the tractions on the surface
\beq
\label{limit}
p_i^+(\bx_0) = \lim_{\bx \to \bx_0} \int_S n_j^+(\bx) K_{ijk}(\bx,\bxi) \Delta u_k(\bxi) d\bxi.
\eeq
The described procedure works when the tractions on the discontinuity surface $S$ are symmetrical. Moreover, the limit in (\ref{limit}) is not straightforward and due 
care should be taken to ensure the convergence of the integrals.

In this paper, we consider the three-dimensional problem of a semi-infinite crack at the interface between two dissimilar elastic half-spaces and apply a method which
avoids the use of the Green's function and the tedious limiting procedure. We also do not use the assumption of symmetrical load.

Our results are based on integral transforms and two fundamental notions of linear elasticity: the Betti reciprocal theorem, on the one hand, and the weight function 
approach, on the other. 
The Betti identity has been extensively used in the perturbation analysis of two and three-dimensional cracks (Willis and Movchan, 1995; Bercial-Velez et al., 2005).
The concept of weight function was introduced in linear fracture mechanics by Bueckner (1970) for two-dimensional problem and later extended to general
three-dimensional problems by Willis and Movchan (1995), Piccolroaz et al. (2007).

We consider a semi-infinite crack at the interface between two dissimilar elastic half-spaces $\Reals^3_\pm = \{ \bx = (x_1,x_2,x_3) \in \Reals^3: \pm x_2 >0 \}$,
see Fig.~\ref{fig01}. The two elastic half-spaces are assumed to be isotropic with shear modulus and Poisson's ratio denoted by $\mu_\pm$ and $\nu_\pm$,
respectively. The crack lies on the half-plane $\Reals^2_- = \{ \bx = (x_1,x_2,x_3) \in \Reals^3: x_1 < 0, x_2 = 0 \}$. The crack faces are loaded by a system of,
not necessarily symmetrical, distributed forces $\bp_\pm$.
\begin{figure}[!htcb]
\centering
\includegraphics[width=10cm]{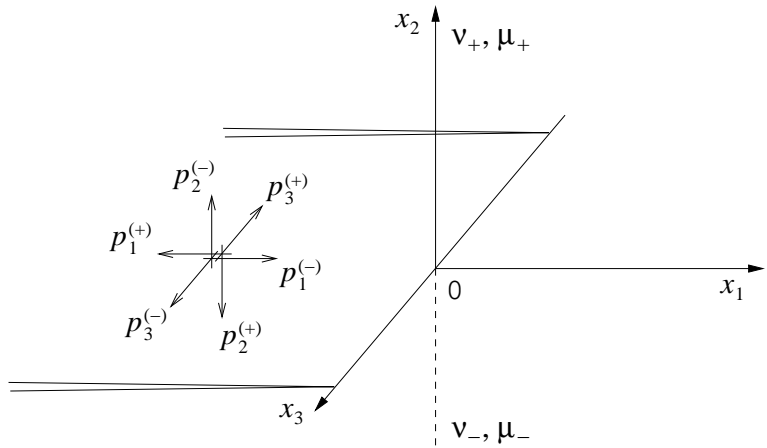}
\caption{\footnotesize Semi-infinite interfacial crack loaded by asymmetrical forces on the crack faces.}
\label{fig01}
\end{figure}

\noindent
It is convenient to introduce the symmetrical and skew-symmetrical parts of the loading as follows
\beq
\langle \bp \rangle = \frac{1}{2}(\bp_+ + \bp_-), \quad \jump{0.15}{\bp} = \bp_+ - \bp_-,
\eeq
where we used standard notations to denote the average, $\langle f \rangle$, and the jump, $\jump{0.15}{f}$, of a function $f$ across the plane containing the crack,
$x_2 = 0$,
\beq
\langle f \rangle(x_1,x_3) = \frac{1}{2} [f(x_1,0^+,x_3) + f(x_1,0^-,x_3)], \quad
\jump{0.15}{f}(x_1,x_3) = f(x_1,0^+,x_3) - f(x_1,0^-,x_3).
\eeq
In Section \ref{sec2}, we introduce the fundamental reciprocal identity and the concept of weight function, as special singular solution to the homogeneous crack
problem. These preliminary results are then used in the sequel for the analysis of both two and three-dimensional elastic problems. In particular, the two-dimensional
problem is analysed in Section \ref{secmode3} (Mode III) and Section \ref{secmode12} (Mode I and II), while the general three-dimensional probelm is analysed in
Section \ref{sec3d}.

\section{Preliminary Results. The Betti formula and weight functions}
\label{sec2}

In this section, we extend the Betti formula to the case of general asymmetrical loading applied at the crack surfaces. The Betti formula is used in order to relate
the physical solution to the weight function, which is a special singular solution to the homogeneous problem (traction-free crack faces) (Willis and Movchan, 1995;
Piccolroaz et al., 2007). For the sake of generality, we allow for traction and displacement jumps across the interface, so that results are general and applicable also
to the case of imperfect interfaces.

In the absence of body forces, the Betti formula takes the form
\beq
\int_{\partial\Omega} \{ \bsigma^{(1)} \bn \cdot \bu^{(2)} - \bsigma^{(2)} \bn \cdot \bu^{(1)} \} ds = 0,
\eeq
where $\partial\Omega$ is any surface enclosing a region $\Omega$ within which both displacement fields $\bu^{(1)}$ and $\bu^{(2)}$ satisfy the
equations of equilibrium, with corresponding stress states $\bsigma^{(1)}$ and $\bsigma^{(2)}$, and $\bn$ denotes the outward normal to
$\partial\Omega$.

Applying the Betti formula to a hemispherical domain in the upper half-space $\Reals_+^3$, whose plane boundary is $x_2 = 0^+$ and whose radius
$R$ will be allowed to tend to infinity, we obtain, in the limit $R \to \infty$,
\beq
\label{betti}
\int\limits_{(x_2=0^+)}
\{ \bsigma_2^{(1)}(x_1,0^+,x_3) \cdot \bu^{(2)}(x_1,0^+,x_3) - \bsigma_2^{(2)}(x_1,0^+,x_3) \cdot \bu^{(1)}(x_1,0^+,x_3) \} dx_1 dx_3 = 0,
\eeq
provided that the fields $\bu^{(1)}$ and $\bu^{(2)}$ decay suitably fast at infinity. The notation $\bsigma_2$ is used to denote the traction
vector acting on the plane $x_2 = 0$: $\bsigma_2 = \bsigma \be_2$.

We can assume that $\bu^{(1)}$ represents the physical field associated with the crack loaded at its surface, whereas $\bu^{(2)}$ represents a
non-trivial solution of the homogeneous problem, the so-called weight function, defined as follows
\beq
\lb{trans}
\bu^{(2)}(x_1,x_2,x_3) = \bR \bU(-x_1,x_2,-x_3),
\eeq
where $\bR$ is a rotation matrix
\beq
\bR =
\left[
\barr{ccc}
-1 & 0 &  0 \\
 0 & 1 &  0 \\
 0 & 0 & -1
\earr
\right].
\eeq
Note that the transformation (\ref{trans}) corresponds to introducing a change of coordinates in the solution $\bu^{(2)}$, namely a
rotation about the $x_2$-axis through an angle $\pi$. It is straightforward to verify that the weight function $\bU$ satisfies the equations of
equilibrium, but in a different domain where the crack is placed along the semi-plane $x_2 = 0$, $x_1 >0$.
The notation $\bSigma$ will be used for components of stress corresponding to the displacement field $\bU$,
\beq
\bsigma^{(2)}(x_1,x_2,x_3) = \bR \bSigma(-x_1,x_2,-x_3) \bR.
\eeq
Replacing $\bu^{(2)}(x_1, x_2, x_3)$ with $\bu^{(2)}(x_1 - x_1', x_2, x_3 - x_3')$, which corresponds to a shift within the plane $(x_1, x_3),$
we obtain
\beq
\label{eq20}
\int\limits_{(x_2=0^+)}
\{ \bR \bU(x_1'-x_1,0^+,x_3'-x_3) \cdot \bsigma_2(x_1,0^+,x_3)
- \bR \bSigma_2(x_1'-x_1,0^+,x_3'-x_3) \cdot \bu(x_1,0^+,x_3) \} dx_1 dx_3 = 0.
\eeq
A similar equation can be derived by applying the Betti formula to a hemispherical domain in the lower half-space $\Reals_-^3$,
\beq
\label{eq21}
\int\limits_{(x_2=0^-)}
\{ \bR \bU(x_1'-x_1,0^-,x_3'-x_3) \cdot \bsigma_2(x_1,0^-,x_3)
- \bR \bSigma_2(x_1'-x_1,0^-,x_3'-x_3) \cdot \bu(x_1,0^-,x_3) \} dx_1 dx_3 = 0.
\eeq
Subtracting (\ref{eq21}) from (\ref{eq20}), we obtain
\beq
\label{recid}
\barr{r}
\ds
\int\limits_{(x_2=0)}
\left\{ \bR \jump{0.15}{\bU}(x_1'-x_1,x_3'-x_3) \cdot \langle \bsigma_2 \rangle(x_1,x_3) +
\bR \langle \bU \rangle(x_1'-x_1,x_3'-x_3) \cdot \jump{0.15}{\bsigma_2}(x_1,x_3) - \right. \\[3mm]
\ds
\left. - \bR \langle \bSigma_2 \rangle(x_1'-x_1,x_3'-x_3) \cdot \jump{0.15}{\bu}(x_1,x_3) \right\} dx_1 dx_3 = 0.
\earr
\eeq
Let us introduce the notations
\beq
\label{supers}
f^{(+)}(x_1,x_3) = f(x_1,x_3) H(x_1), \quad f^{(-)}(x_1,x_3) = f(x_1,x_3) H(-x_1),
\eeq
where $H$ denotes the Heaviside function, so that
$$
f(x_1,x_3) = f^{(+)}(x_1,x_3) + f^{(-)}(x_1,x_3).
$$
The reciprocity identity (\ref{recid}) becomes
\beq
\label{recid2}
\barr{r}
\ds
\int_{-\infty}^{\infty} \left\{
\bR \jump{0.15}{\bU}(x_1'-x_1,x_3'-x_3) \cdot \langle \bsigma_2 \rangle^{(+)}(x_1,x_3)
- \bR \langle \bSigma_2 \rangle(x_1'-x_1,x_3'-x_3) \cdot \jump{0.15}{\bu}^{(-)}(x_1,x_3)
\right\} dx_1 dx_3 = \\[5mm]
\ds
- \int_{-\infty}^{\infty} \left\{
\bR \jump{0.15}{\bU}(x_1'-x_1,x_3'-x_3) \cdot \langle \bsigma_2 \rangle^{(-)}(x_1,x_3)
+ \bR \langle \bU \rangle(x_1'-x_1,x_3'-x_3) \cdot \jump{0.15}{\bsigma_2}^{(-)}(x_1,x_3)
\right\} dx_1 dx_3 \\[5mm]
\ds
- \int_{-\infty}^{\infty} \left\{
\bR \langle \bU \rangle(x_1'-x_1,x_3'-x_3) \cdot \jump{0.15}{\bsigma_2}^{(+)}(x_1,x_3)
- \bR \langle \bSigma_2 \rangle(x_1'-x_1,x_3'-x_3) \cdot \jump{0.15}{\bu}^{(+)}(x_1,x_3)
\right\} dx_1 dx_3.
\earr
\eeq

The identity (\ref{recid2}) can be written in an equivalent form using the convolution with respect to both $x_1$ and $x_3$ variables, denoted by the symbol
$\circledast$,
\beq
\label{recid3}
\bR \jump{0.15}{\bU} \circledast \langle \bsigma_2 \rangle^{(+)} - \bR \langle \bSigma_2 \rangle \circledast \jump{0.15}{\bu}^{(-)} =
- \bR \jump{0.15}{\bU} \circledast \langle \bsigma_2 \rangle^{(-)} - \bR \langle \bU \rangle \circledast \jump{0.15}{\bsigma_2}^{(-)}
- \bR \langle \bU \rangle \circledast \jump{0.15}{\bsigma_2}^{(+)} + \bR \langle \bSigma_2 \rangle \circledast \jump{0.15}{\bu}^{(+)}.
\eeq

Note that the superscripts $^{(+)}$ and $^{(-)}$ denote functions whose support is restricted to the positive and negative semi-axes, respectively. Thus,
$\langle \bsigma_2 \rangle^{(+)}$ is the traction along the interface, ahead of the crack tip, whereas $\jump{0.15}{\bu}^{(-)}$ is the crack opening (displacement
discontinuity across the crack faces).

In the sequel, we will assume perfect contact conditions at the interface, so that traction and displacement are continuous for $x_1 > 0$:
$\jump{0.15}{\bsigma_2}^{(+)} = \jump{0.15}{\bu}^{(+)} = 0$. Finally, the symmetric and skew-symmetric loads applied on the crack faces will be denoted by
$\langle \bsigma_2 \rangle^{(-)} = \langle \bp \rangle$ and $\jump{0.15}{\bsigma_2}^{(-)} = \jump{0.15}{\bp}$, respectively.

The terms $\jump{0.15}{\bU}$ and $\langle \bU \rangle$ are known as the symmetric and skew-symmetric weight functions. The notion of skew-symmetric weight function
was introduced in Piccolroaz et al.\ (2009) for problems of cracks loaded by asymmetrical forces. The term $\langle \bSigma_2 \rangle$ stands for the traction along
the plane $x_2 = 0$, corresponding to the singular solution $\bU$.

\section{Mode III}
\label{secmode3}

In the case of antiplane deformation, the Betti formula (\ref{recid3}) relating the physical field $u_3$, $\sigma_{23}$ with
the weight function $U_3$, $\Sigma_{23}$ reduces to the scalar equation
\beq
\label{id1}
\jump{0.15}{U_3} * \langle \sigma_{23} \rangle^{(+)} - \langle \Sigma_{23} \rangle * \jump{0.15}{u_3}^{(-)} =
-\jump{0.15}{U_3} * \langle p_3 \rangle - \langle U_3 \rangle * \jump{0.15}{p_3},
\eeq
where the symbol $*$ denotes the convolution with respect to the variable $x_1$.

Let us introduce the Fourier transform with respect to the variable $x_1$ as follows
\beq
\tilde{f}(\beta) = \mF[f(x_1)] = \int_{-\infty}^{\infty} f(x_1) e^{i\beta x_1} dx_1.
\eeq
Applying the Fourier transform to the identity (\ref{id1}) we get
\beq
\label{id2}
\jump{0.15}{\tilde{U}_3} \langle \tilde{\sigma}_{23} \rangle^+ - \langle \tilde{\Sigma}_{23} \rangle \jump{0.15}{\tilde{u}_3}^- =
-\jump{0.15}{\tilde{U}_3} \langle \tilde{p}_3 \rangle - \langle \tilde{U}_3 \rangle \jump{0.15}{\tilde{p}_3},
\eeq
where the superscripts $^+$ and $^-$ denote functions analytic in the upper and lower half-planes, respectively.

We can now multiply both sides of (\ref{id2}) by $\jump{0.15}{\tilde{U}_3}^{-1}$ to obtain
\beq
\label{id2a}
\langle \tilde{\sigma}_{23} \rangle^+ - B \jump{0.15}{\tilde{u}_3}^- =
- \langle \tilde{p}_3 \rangle - A \jump{0.15}{\tilde{p}_3},
\eeq
where the factors in front of unknown functions are given by
\beq
A = \jump{0.15}{\tilde{U}_3}^{-1} \langle \tilde{U}_3 \rangle, \quad
B = \jump{0.15}{\tilde{U}_3}^{-1} \langle \tilde{\Sigma}_{23} \rangle.
\eeq
They can be computed from the relationships, which hold for the symmetric and skew-symmetric weight functions (Piccolroaz et al., 2009)
\beq
\jump{0.15}{\tilde{U}_3} = -\frac{b + e}{|\beta|} \langle \tilde{\Sigma}_{23} \rangle, \quad
\langle \tilde{U}_3 \rangle = \frac{\eta}{2} \jump{0.15}{\tilde{U}_3}.
\eeq
\noindent
Thus, we can easily obtain
\beq
A = \frac{\eta}{2}, \quad B = - \frac{|\beta|}{b + e},
\eeq
where $b$, $e$ and $\eta$ are the following bimaterial constants,
\beq
\label{parb}
b = \frac{1 - \nu_+}{\mu_+} + \frac{1 - \nu_-}{\mu_-}, \quad
e = \frac{\nu_+}{\mu_+} + \frac{\nu_-}{\mu_-}, \quad
\eta = \frac{\mu_- - \mu_+}{\mu_- + \mu_+}.
\eeq

If we apply the inverse Fourier transform to (\ref{id2a}), we obtain for the two opposite cases $x_1 < 0$ and $x_1 > 0$ the following relationships:
\beq
\label{id3}
\langle p_3 \rangle + \mF^{-1}_{x_1 < 0}\Big[A \jump{0.15}{\tilde{p}_3}\Big] =
\mF^{-1}_{x_1 < 0}\Big[B \jump{0.15}{\tilde{u}_3}^-\Big],
\eeq
\beq
\label{id4}
\langle \sigma_{23} \rangle^{(+)} = \mF^{-1}_{x_1 > 0}\Big[B \jump{0.15}{\tilde{u}_3}^-\Big],
\eeq
Note that the term $\langle \tilde{\sigma}_{23} \rangle^+$ in (\ref{id2}) cancels from (\ref{id3}) because it is a ``$+$'' function, while the terms
$\langle \tilde{p}_3 \rangle$ and $\jump{0.15}{\tilde{p}_3}$ cancel from (\ref{id4}) because they are ``$-$'' functions.

Finally, we find the inverse Fourier transform of the function $|\beta| \jump{0.15}{\tilde{u}_3}^-$. To this purpose we write
$|\beta| \jump{0.15}{\tilde{u}_3}^- = \sign(\beta) \cdot \beta\jump{0.15}{\tilde{u}_3}^-$ and observe that $\sign(\beta)$ is the Fourier transform of
$-i/(\pi x_1)$, whereas $\beta\jump{0.15}{\tilde{u}_3}^-$ is the Fourier transform of $i \partial \jump{0.15}{u_3}^- / \partial x_1$. Therefore we get
\beq
\label{not}
\mF^{-1}\Big[|\beta| \jump{0.15}{\tilde{u}_3}^-\Big] = \frac{1}{\pi x_1} * \frac{\partial \jump{0.15}{u_3}^-}{\partial x_1}.
\eeq
Note that, with slight abuse of notation, we write the convolution of two functions $f(x_1)$ and $g(x_1)$ in the form
\beq
\int_{-\infty}^{\infty} f(x_1 - \xi) g(\xi) d\xi = f(x_1) * g(x_1),
\eeq
so that the right-hand side in (\ref{not}) has to be interpreted as follows
\beq
\frac{1}{\pi x_1} * \frac{\partial \jump{0.15}{u_3}^-}{\partial x_1} =
\frac{1}{\pi} \int_{-\infty}^{\infty}\frac{1}{x_1 - \xi} \frac{\partial \jump{0.15}{u_3}^-}{\partial \xi} d\xi.
\eeq

\noindent
To simplify notations, we introduce the singular operator $\mS$ and the orthogonal projectors $\mP_\pm$ ($\mP_+ + \mP_- = \mI$) acting on the real axis:
\beq
\psi = \mS \varphi =
\frac{1}{\pi x_1} * \varphi(x_1) =
\frac{1}{\pi} \int_{-\infty}^{\infty}\frac{\varphi(\xi)}{x_1 - \xi}  d\xi,
\eeq
\beq
\mP_\pm \varphi =
\left\{
\barr{ll}
\varphi(x_1) & \pm x_1 \ge 0, \\[3mm]
0 & \text{otherwise}.
\earr
\right.
\eeq
Note that the operator $\mS$ is the classic singular operator of the Cauchy type; it transforms any function $\varphi$, satisfying the H\"{o}lder condition,
into a new function $\mS \varphi$ which also satisfies the H\"{o}lder condition (Muskhelishvili, 1946). The properties of the operator $\mS$ in other
functional spaces have been discussed in Pr\"{o}ssdorf (1974).

The integral identities (\ref{id3}) and (\ref{id4}) for Mode III deformation become
\beq
\label{res1}
\langle p_3 \rangle + \frac{\eta}{2} \jump{0.15}{p_3} =
-\frac{1}{b + e} \mS^{(s)} \frac{\partial \jump{0.15}{u_3}^{(-)}}{\partial x_1}, \quad x_1 < 0,
\eeq
\beq
\label{res2}
\langle \sigma_{23} \rangle^{(+)} =
-\frac{1}{b + e} \mS^{(c)} \frac{\partial \jump{0.15}{u_3}^{(-)}}{\partial x_1}, \quad x_1 > 0,
\eeq
where $\mS^{(s)} = \mP_- \mS \mP_-$ is a singular operator, whereas $\mS^{(c)} = \mP_+ \mS \mP_-$ is a compact operator (see Gakhov and Cherski 1978; Krein, 1958;
Gohberg and Krein, 1958).
Despite the fact that these operators look similar at a first glance, they are essentially different. Indeed, $\mS^{(s)}: F(\Reals_-) \to F(\Reals_-)$,
while $\mS^{(c)}: F(\Reals_-) \to F(\Reals_+)$, where $F(\Reals_\pm)$ is some functional space of functions defined on $\Reals_\pm$.

To make this point clear, let us write eqs. (\ref{res1}) and (\ref{res2}) in extended form as
\beq
\label{res1b}
\langle p_3 \rangle + \frac{\eta}{2} \jump{0.15}{p_3} =
-\frac{1}{\pi(b + e)} \int_{-\infty}^{0} \frac{1}{x_1 - \xi} \frac{\partial \jump{0.15}{u_3}^{(-)}}{\partial \xi} d\xi, \quad x_1 < 0,
\eeq
\beq
\label{res2b}
\langle \sigma_{23} \rangle^{(+)} =
-\frac{1}{\pi(b + e)} \int_{-\infty}^{0} \frac{1}{x_1 - \xi} \frac{\partial \jump{0.15}{u_3}^{(-)}}{\partial \xi} d\xi, \quad x_1 > 0.
\eeq
The integral in (\ref{res1b}) is a Cauchy singular integral with a moving singularity, whereas the integral in (\ref{res2b}) is an integral
with a fixed point singularity.

Note that, in the particular case of a homogeneous body, we have $\eta = 0$ and $b + e = 2/\mu$, so that the integral identities (\ref{res1}) and (\ref{res2})
reduce to
\beq
\label{mode_III}
\langle p_3 \rangle =
-\frac{\mu}{2} \mS^{(s)} \frac{\partial \jump{0.15}{u_3}^{(-)}}{\partial x_1}, \quad x_1 < 0,
\eeq
\beq
\langle \sigma_{23} \rangle^{(+)} =
-\frac{\mu}{2} \mS^{(c)} \frac{\partial \jump{0.15}{u_3}^{(-)}}{\partial x_1}, \quad x_1 > 0.
\eeq

The two equations (\ref{res1b}) and (\ref{res2b}) form the system of integral identities for the Mode III deformation. The first equation (\ref{res1b}) provides
the integral relationship between the applied loading $\langle p_3 \rangle$, $\jump{0.15}{p_3}$ and the resulting crack opening $\jump{0.15}{u_3}^{(-)}$. This is
a singular integral equation and it is, generally speaking, invertible. However, the inverse operator depends on the properties of the solution.
The second equation (\ref{res2b}) can be considered as an additional equation which allows to define the proper behaviour of the solution $\jump{0.15}{u_3}^{(-)}$
and also, after the first equation has been inverted, to evaluate the traction ahead of the crack tip $\langle \sigma_{23} \rangle^{(+)}$. The reason for this is that
the operator in the right-hand side of (\ref{res2b}) is a compact one and thus it is not invertible.

More details on the theory of singular integral equations can be found in Muskhelishvili (1946). The inversion of the singular operator $\mS^{(s)}$, in some
specific cases of interest for fracture mechanics and, in particular, for hydraulic fracturing, is discussed in Appendix \ref{appS}.

\section{Mode I and II}
\label{secmode12}

In the case of plane strain deformation, the Betti identity (\ref{recid3}) relating the physical solution $\bu = [u_1,u_2]^T$,
$\bsigma_2 = [\sigma_{21},\sigma_{22}]^T$ with the weight function $\bU$, $\bSigma_2$ is given by
\beq
\label{id1b}
\bR \jump{0.15}{\bU} * \langle \bsigma_2 \rangle^{(+)} - \bR \langle \bSigma_2 \rangle * \jump{0.15}{\bu}^{(-)} =
- \bR \jump{0.15}{\bU} * \langle \bp \rangle - \bR \langle \bU \rangle * \jump{0.15}{\bp}.
\eeq
where $\langle \bp \rangle = [\langle p_1 \rangle, \langle p_2 \rangle]^T$, $\jump{0.15}{\bp} = [\jump{0.15}{p_1}, \jump{0.15}{p_2}]^T$ are the symmetric and
skew-symmetric parts of the loading. Again we assume ideal contact conditions, so that traction and displacement are continuous across the interface:
$\jump{0.15}{\bsigma_{2}}^{(+)} = \jump{0.15}{\bu}^{(+)} = 0$. Here and in the sequel of this section, we use the following matrices:
\beq
\label{Rmatrix}
\bR =
\left[
\barr{cc}
-1 & 0 \\
0 & 1
\earr
\right], \quad
\bI =
\left[
\barr{cc}
1 & 0 \\
0 & 1
\earr
\right], \quad
\bE =
\left[
\barr{cc}
0 & 1 \\
-1 & 0
\earr
\right].
\eeq
Note that the symmetric and skew-symmetric weight functions, $\jump{0.15}{\bU}$ and $\langle \bU \rangle$, and the corresponding traction
$\langle \bSigma_2 \rangle$ are represented by 2$\times$2 matrices. In fact, in the case of an elastic bimaterial plane, there are two linearly independent
weight functions, $\bU^j = [U_1^j,U_2^j]^T$, $\bSigma_2^j = [\Sigma_{21}^j,\Sigma_{22}^j]^T$, $j = 1,2$, and it is possible to construct the weight function
tensors by ordering the components of each weight function in columns of 2$\times$2 matrices (Piccolroaz et al., 2009)
\beq
\bU =
\left[
\barr{cc}
U_1^1 & U_1^2 \\[1mm]
U_2^1 & U_2^2
\earr
\right], \quad
\bSigma_2 =
\left[
\barr{cc}
\Sigma_{21}^1 & \Sigma_{21}^2 \\[1mm]
\Sigma_{22}^1 & \Sigma_{22}^2
\earr
\right].
\eeq
Applying the Fourier transform to the equation (\ref{id1b}), we obtain
\beq
\label{id2b}
\jump{0.15}{\tilde{\bU}}^T \bR \langle \tilde{\bsigma}_2 \rangle^+ - \langle \tilde{\bSigma}_2 \rangle^T \bR \jump{0.15}{\tilde{\bu}}^- =
-\jump{0.15}{\tilde{\bU}}^T \bR \langle \tilde{\bp} \rangle - \langle \tilde{\bU} \rangle^T \bR \jump{0.15}{\tilde{\bp}}.
\eeq
Multiplying both sides by $\bR^{-1} \jump{0.15}{\tilde{\bU}}^{-T}$ we get
\beq
\label{id2ba}
\langle \tilde{\bsigma}_2 \rangle^+ - \bB \jump{0.15}{\tilde{\bu}}^- =
-\langle \tilde{\bp} \rangle - \bA \jump{0.15}{\tilde{\bp}},
\eeq
where $\bA$ and $\bB$ are the following matrices
\beq
\bA = \bR^{-1} \jump{0.15}{\tilde{\bU}}^{-T} \langle \tilde{\bU} \rangle^T \bR, \quad
\bB = \bR^{-1} \jump{0.15}{\tilde{\bU}}^{-T} \langle \tilde{\bSigma}_2 \rangle^T \bR,
\eeq
which can be computed using results for the symmetric and skew-symmetric weight functions obtained in Antipov (1999) and Piccolroaz et al. (2009). Namely
\beq
\jump{0.15}{\tilde{\bU}} = -\frac{1}{|\beta|} \big[b \bI - id\sign(\beta) \bE\big] \langle \tilde{\bSigma}_2 \rangle, \quad
\langle \tilde{\bU} \rangle = - \frac{b}{2|\beta|} \big[\alpha \bI - i\gamma\sign(\beta) \bE\big] \langle \tilde{\bSigma}_2 \rangle,
\eeq
where $b$ is defined in (\ref{parb}), $d/b$ and $\alpha$ are the so-called Dundurs parameters, while $\gamma$ is another bimaterial constant:
\beq
d = \frac{1 - 2\nu_+}{2\mu_+} - \frac{1 - 2\nu_-}{2\mu_-},
\eeq
\beq
\alpha = \frac{\mu_-(1 - \nu_+) - \mu_+(1 - \nu_-)}{\mu_-(1 - \nu_+) + \mu_+(1 - \nu_-)}, \quad
\gamma = \frac{\mu_-(1 - 2\nu_+) + \mu_+(1 - 2\nu_-)}{2\mu_-(1 - \nu_+) + 2\mu_+(1 - \nu_-)}.
\eeq
As a result, we obtain
\beq
\bA =
\frac{b}{2(b^2 - d^2)}
\big[
(b\alpha - d\gamma) \bI + i(d\alpha - b\gamma)\sign(\beta) \bE
\big],
\eeq
\beq
\bB =
-\frac{|\beta|}{b^2 - d^2}
\big[
b \bI + id\sign(\beta) \bE
\big].
\eeq
Inverting the Fourier transform for the two cases $x_1 < 0$ and $x_1 > 0$, we get
\beq
\label{id3b}
\langle \bp \rangle + \mF^{-1}_{x_1 < 0}\Big[\bA \jump{0.15}{\tilde{\bp}}\Big] =
\mF^{-1}_{x_1 < 0}\Big[\bB \jump{0.15}{\tilde{\bu}}^-\Big],
\eeq
\beq
\label{id4b}
\langle \bsigma_2 \rangle^{(+)} + \mF^{-1}_{x_1 > 0}\Big[\bA \jump{0.15}{\tilde{\bp}}\Big] =
\mF^{-1}_{x_1 > 0}\Big[\bB \jump{0.15}{\tilde{\bu}}^-\Big].
\eeq
Similarly to the previous section, the term $\langle \tilde{\bsigma} \rangle^+$ in (\ref{id2b}) cancels from (\ref{id3b}) because it is a ``$+$'' function,
whereas the term $\langle \tilde{\bp} \rangle$ cancels from (\ref{id4b}) because it is a ``$-$'' function.

To proceed further, we need to invert functions of the form $\sign(\beta) \jump{0.15}{\tilde{p}_i}$, $|\beta| \jump{0.15}{\tilde{u}_i}$ and
$\beta \jump{0.15}{\tilde{u}_i}$. This is done in a way similar to the previous section and we finally obtain the integral identities for plane strain
deformation as follows
\beq
\label{general_2D}
\langle \bp \rangle + \bmA^{(s)} \jump{0.15}{\bp} = \bmB^{(s)} \frac{\partial \jump{0.15}{\bu}^{(-)}}{\partial x_1}, \quad x_1 < 0,
\eeq
\beq
\label{general_2D2}
\langle \bsigma_2 \rangle^{(+)} + \bmA^{(c)} \jump{0.15}{\bp} = \bmB^{(c)} \frac{\partial \jump{0.15}{\bu}^{(-)}}{\partial x_1}, \quad x_1 > 0,
\eeq
where $\bmA^{(s)}$, $\bmB^{(s)}: F(\Reals_-) \to F(\Reals_-)$, and $\bmA^{(c)}, \bmB^{(c)}: F(\Reals_-) \to F(\Reals_+)$ are the following matrix operators
\beq
\label{op1}
\bmA^{(s)} =
\frac{b}{2(b^2 - d^2)}
\left[
(b\alpha - d\gamma) \bI + (d\alpha - b\gamma) \bE \mS^{(s)}
\right], \quad
\bmB^{(s)} =
-\frac{1}{b^2 - d^2}
\left[
b \bI \mS^{(s)} - d \bE
\right].
\eeq
\beq
\label{op2}
\bmA^{(c)} =
\frac{b(d\alpha - b\gamma)}{2(b^2 - d^2)} \bE \mS^{(c)}, \quad
\bmB^{(c)} =
-\frac{b}{b^2 - d^2} \bI \mS^{(c)}.
\eeq
Equations (\ref{general_2D}) and (\ref{general_2D2}), together with the definition of operators (\ref{op1}) and (\ref{op2}), form the system of
integral identities for Mode I/II deformation. In particular, (\ref{general_2D}) is a system of two singular integral equations and this system is coupled,
in general.

Note that, in the case where the Dundurs parameter $d$ vanishes, the operators can be written as
\beq
\label{homo_2D_d}
\bmA^{(s)} =
\frac{\alpha}{2} \bI -\frac{\gamma}{2} \bE \mS^{(s)}, \quad
\bmB^{(s)} = -\frac{1}{b} \bI \mS^{(s)}.
\eeq
\beq
\label{homo_2D_d2}
\bmA^{(c)} = -\frac{\gamma}{2} \bE \mS^{(c)}, \quad
\bmB^{(c)} = -\frac{1}{b} \bI \mS^{(c)}.
\eeq
This means that the system (\ref{general_2D}) decouples in the case $d = 0$ and reduces to
\beq
\label{homo_2Dd}
-\frac{1}{b} \mS^{(s)} \frac{\partial \jump{0.15}{\bu}^{(-)}}{\partial x_1} =
\langle \bp \rangle + \left(\frac{\alpha}{2}\bI - \frac{\gamma}{2} \bE \mS^{(s)}\right) \jump{0.15}{\bp}, \quad x_1 < 0.
\eeq
In the particular case of a homogeneous material, we additionally have $\alpha = 0$, $b = 2(1 - \nu)/\mu$ and $\gamma = (1 - 2\nu)/[2(1 - \nu)]$, so that the matrix
operators simplify further to the form:
\beq
\label{homo_2D}
\bmA^{(s,c)} = -\frac{(1 - 2\nu)}{4(1 - \nu)} \bE \mS^{(s,c)}, \quad
\bmB^{(s,c)} = -\frac{\mu}{2(1 - \nu)} \bI \mS^{(s,c)}.
\eeq
and the singular integral equation (\ref{homo_2Dd}) simplifies to
\beq
\label{homo_2Da}
-\frac{\mu}{2} \mS^{(s)} \frac{\partial \jump{0.15}{\bu}^{(-)}}{\partial x_1} =
(1 - \nu)\langle \bp \rangle - \frac{1 - 2\nu}{4} \bE \mS^{(s)} \jump{0.15}{\bp}, \quad x_1 < 0.
\eeq
In the case of symmetrical load ($\jump{0.15}{\bp} = 0$) this formula can be found for example in Rice (1968). It differs from the Mode III case 
(\ref{mode_III}) only for the coefficient in front of $\langle \bp \rangle$. If in addition $\nu = 0$, then the two formulas fully coincide, altough one is scalar and
the other is vector one.

The solution of either equation (\ref{homo_2Dd}) or (\ref{homo_2Da}), allowing for the calculation of the crack opening related to a given load applied on the crack
faces, requires the inversion of the operator $S^{(s)}$, which is discussed in details in Appendix \ref{appS}.
To invert the operator (\ref{general_2D}) in the general case, one needs to use results from Vekua (1970) where the systems of singular integral equations
are discussed and analysed thoroughly.

However, we propose here a simple way to decouple the system (\ref{general_2D}) by acting on it by the following operator:
\beq
\bmR = b \bI \mS^{(s)} + d\bE.
\eeq
As a result, equation (\ref{general_2D}) decouples as follows
\beq
\label{2D_decouple}
(1 - p\mK) \frac{\partial \jump{0.15}{\bu}^{(-)}}{\partial x_1} =
\frac{1}{2} \left\{ b\alpha \bI \mS^{(s)} + \bE \left[ b\gamma + p (d\alpha - b\gamma) \mK \right] \right\} \jump{0.15}{\bp}
+ \left( b \bI \mS^{(s)} + d\bE \right) \langle \bp \rangle,
\eeq
where we have introduced a new bimaterial constant ($0 < p \le 1$)
\beq
p = \frac{b^2}{b^2 - d^2}.
\eeq
In deriving (\ref{2D_decouple}) we used results from Muskhelishvili (1946) to prove that
\beq
\left(\mS^{(s)}\right)^2 = -\mI + \mK,
\eeq
where $\mI$ is the identity operator and $\mK: F(\Reals_-) \to F(\Reals_-)$ is a compact operator given by:
\beq
\mK\varphi = \frac{1}{\pi^2} \int_{-\infty}^{0} \frac{\log|\xi/x_1|}{x_1 - \xi} \varphi(\xi) d\xi.
\eeq
It is easy to see that in the case of a homogeneous body ($d = 0$, $\alpha = 0$ and thus $p = 1$) the equation (\ref{2D_decouple}) is equivalent to
(\ref{homo_2Da}). In the general case, the operator $1 - p\mK$ can be analytically inverted again using results from Muskhelishvili (1946). Alternatively, the
Fredholm equation (\ref{2D_decouple}) can be easily inverted by numerical techniques.

\section{General 3D case}
\label{sec3d}

For the case of an elastic bimaterial space, there are three linearly independent weight functions, from which we can construct the weight function tensors by
ordering the three components of each weight function in columns of 3$\times$3 matrices (Piccolroaz et al., 2007; Piccolroaz et al., 2011)
\beq
\bU =
\left[
\barr{ccc}
U_1^1 & U_1^2 & U_1^3 \\[1mm]
U_2^1 & U_2^2 & U_2^3 \\[1mm]
U_3^1 & U_3^2 & U_3^3 \\[1mm]
\earr
\right], \quad
\bSigma_2 =
\left[
\barr{ccc}
\Sigma_{21}^1 & \Sigma_{21}^2 & \Sigma_{21}^3 \\[1mm]
\Sigma_{22}^1 & \Sigma_{22}^2 & \Sigma_{22}^3 \\[1mm]
\Sigma_{23}^1 & \Sigma_{23}^2 & \Sigma_{23}^3 \\[1mm]
\earr
\right].
\eeq
Assuming again perfect interface (so that $\jump{0.15}{\bsigma_{2}}^{(+)} = \jump{0.15}{\bu}^{(+)} = 0$), the Betti identity, relating the physical solution
$\bu = [u_1,u_2,u_3]^T$, $\bsigma_2 = [\sigma_{21},\sigma_{22},,\sigma_{23}]^T$ with the weight function $\bU$, $\bSigma_2$, writes
\beq
\label{id1c}
\bR \jump{0.15}{\bU} \circledast \langle \bsigma_2 \rangle^{(+)} - \bR \langle \bSigma_2 \rangle \circledast \jump{0.15}{\bu}^{(-)} =
- \bR \jump{0.15}{\bU} \circledast \langle \bp \rangle - \bR \langle \bU \rangle \circledast \jump{0.15}{\bp},
\eeq
where the symbol $\circledast$ denotes convolution with respect to both variables $x_1$ and $x_3$ and $\bR$ is the matrix defined in Sec.~\ref{sec1}.

Let us introduce the Fourier transform with respect to the $x_1$ and $x_3$ coordinates as follows
\beq
\overline{f}(\beta,x_2,\lambda) = \int_{-\infty}^{\infty} \int_{-\infty}^{\infty} f(x_1,x_2,x_3) e^{i\beta x_1 + i\lambda x_3} dx_1 dx_3.
\eeq
Taking the Fourier transform of (\ref{id1c}) we obtain
\beq
\label{id2c}
\jump{0.15}{\overline{\bU}}^T \bR \langle \overline{\bsigma}_2 \rangle^+ - \langle \overline{\bSigma}_2 \rangle^T \bR \jump{0.15}{\overline{\bu}}^- =
-\jump{0.15}{\overline{\bU}}^T \bR \langle \overline{\bp} \rangle - \langle \overline{\bU} \rangle^T \bR \jump{0.15}{\overline{\bp}}.
\eeq
Multiplying both sides of (\ref{id2c}) by $\bR^{-1} \jump{0.15}{\overline{\bU}}^{-T}$ we get
\beq
\langle \overline{\bp} \rangle^- + \bA \jump{0.15}{\overline{\bp}}^- =
-\langle \overline{\bsigma}_2 \rangle^+ + \bB \jump{0.15}{\overline{\bu}}^-.
\eeq
where
\beq
\bA = \bR^{-1} \jump{0.15}{\overline{\bU}}^{-T} \langle \overline{\bU} \rangle^T \bR, \quad
\bB = \bR^{-1} \jump{0.15}{\overline{\bU}}^{-T} \langle \overline{\bSigma}_2 \rangle^T \bR.
\eeq
The two matrices $\bA$ and $\bB$ can be obtained from the following relationships between symmetric and skew-symmetric weight functions
(Antipov, 1999; Piccolroaz et al., 2011)
\beq
\label{syst1}
\jump{0.15}{\overline{\bU}} = \frac{1}{\rho} \bG \langle \overline{\bSigma}_2 \rangle, \quad
\langle \overline{\bU} \rangle = \frac{1}{\rho} \bF \langle \overline{\bSigma}_2 \rangle,
\eeq
where $\rho = \sqrt{\beta^2 + \lambda^2}$,
\beq
\bG(\beta,\lambda) = -\frac{1}{\rho^2} \left( b\rho^2\bI + e\bE_2 + id\rho\bE_3 \right), \quad
\bF(\beta,\lambda) = -\frac{1}{2\rho^2} \left(b\alpha\rho^2\bI + f\bE_2 + ib\gamma\rho\bE_3 \right).
\eeq
\beq
\bE_2(\beta,\lambda) =
\left[
\barr{ccc}
\lambda^2 & 0 & -\beta\lambda \\
0 & 0 & 0 \\
-\beta\lambda & 0 & \beta^2
\earr
\right],
\quad
\bE_3(\beta,\lambda) =
\left[
\barr{ccc}
0 & -\beta & 0 \\
\beta & 0 &\lambda \\
0 & -\lambda & 0
\earr
\right],
\eeq
$\bI$ is the 3$\times$3 identity matrix, $\bE_2\bE_3 = \b0$ and $e$, $f$ are the following bimaterial parameters:
\beq
e = \frac{\nu_+}{\mu_+} + \frac{\nu_-}{\mu_-}, \quad
f = \frac{\nu_+}{\mu_+} - \frac{\nu_-}{\mu_-}.
\eeq
As a result, we obtain
\beq
\barr{l}
\ds \bA(\beta,\lambda) = \frac{1}{2(b^2 - d^2)(b + e)\rho^2} \times \\[5mm]
\ds \Big\{ b(b\alpha - d\gamma)(b + e)\rho^2\bI + \big[ bd\gamma(b + e) - b\alpha(be + d^2) +
(b^2 - d^2)f \big]\bE_2 + ib(b\gamma - d\alpha)(b + e)\rho\bE_3 \Big\},
\earr
\eeq
\beq
\bB(\beta,\lambda) = \frac{1}{(b^2 - d^2)(b + e)\rho} \Big\{ -b(b + e)\rho^2\bI + (be + d^2)\bE_2 + id(b + e)\rho\bE_3 \Big\}.
\eeq
By inversion of the Fourier transform for the two cases $x_1 < 0$ and $x_1 > 0$, we obtain the two formulae
\beq
\langle \bp \rangle + \mF^{-1}_{x_1 < 0}\mF^{-1}_{x_3}\Big[\bA \jump{0.15}{\overline{\bp}}\Big] =
\mF^{-1}_{x_1 < 0}\mF^{-1}_{x_3}\Big[\bB \jump{0.15}{\overline{\bu}}^-\Big],
\eeq
\beq
\langle \bsigma_2 \rangle^{(+)} + \mF^{-1}_{x_1 > 0}\mF^{-1}_{x_3}\Big[\bA \jump{0.15}{\overline{\bp}}\Big] =
\mF^{-1}_{x_1 > 0}\mF^{-1}_{x_3}\Big[\bB \jump{0.15}{\overline{\bu}}^-\Big].
\eeq

We need now to perform the double Fourier inversion. The details are given in Appendix \ref{appdoub}. The final result is given in terms of the
following operators
\beq
\mQ \varphi =
\frac{1}{\pi \sqrt{x_1^2 + x_3^2}} \circledast \varphi(x_1,x_3) =
\frac{1}{\pi} \int_{-\infty}^{\infty} \int_{-\infty}^{\infty} \frac{\varphi(\xi_1,\xi_3)}{\sqrt{(x_1 - \xi_1)^2 + (x_3 - \xi_3)^2}} d\xi_1 d\xi_3,
\eeq
\beq
\mQ_j \varphi =
\frac{x_j}{\pi (x_1^2 + x_3^2)} \circledast \varphi(x_1,x_3) =
\frac{1}{\pi} \int_{-\infty}^{\infty} \int_{-\infty}^{\infty} \frac{(x_j - \xi_j) \varphi(\xi_1,\xi_3)}{(x_1 - \xi_1)^2 + (x_3 - \xi_3)^2} d\xi_1 d\xi_3, \quad
j = 1,3.
\eeq
\beq
\mQ^{(s,c)} = \mP_{(-,+)} \mQ \mP_-, \quad
\mQ_{j}^{(s,c)} = \mP_{(-,+)} \mQ_j \mP_-, \quad j=1,3.
\eeq
The integral identities for 3D elasticity are given by
\beq
\label{intid3d1}
\langle \bp \rangle + \bmA^{(s)} \jump{0.15}{\bp} = \bmB^{(s)} \jump{0.15}{\bu}^{(-)}, \quad x_1 < 0,
\eeq
\beq
\label{intid3d2}
\langle \bsigma_2 \rangle^{(+)} + \bmA^{(c)} \jump{0.15}{\bp} = \bmB^{(c)} \jump{0.15}{\bu}^{(-)}, \quad x_1 >0,
\eeq
where
\beq
\barr{l}
\ds \bmA^{(s)} = \frac{1}{2(b^2 - d^2)(b + e)} \times \\[3mm]
\ds
\left\{
b(b\alpha - d\gamma)(b + e) \bI + \frac{1}{2} \big[bd\gamma(b + e) - b\alpha(be + d^2) + (b^2 - d^2)f\big] \bmE_2^{(s)} \right. \\[3mm]
\ds \hspace{60mm} \left. -\frac{1}{2} b(b\gamma - d\alpha)(b + e) \mQ^{(s)}\bE_3\left(\frac{\partial}{\partial x_1},\frac{\partial}{\partial x_3}\right)
\right\},
\earr
\eeq
\vspace{5mm}
\beq
\barr{l}
\ds \bmB^{(s)} = \frac{1}{2(b^2 - d^2)(b + e)} \times \\[5mm]
\ds \left\{ b(b + e) \bI \mQ^{(s)} \left(\frac{\partial^2}{\partial x_1^2} + \frac{\partial^2}{\partial x_3^2}\right) -
(be + d^2) \mQ^{(s)} \bE_2\left(\frac{\partial}{\partial x_1},\frac{\partial}{\partial x_3}\right) -
2d(b + e) \bE_3\left(\frac{\partial}{\partial x_1},\frac{\partial}{\partial x_3}\right) \right\},
\earr
\eeq
\vspace{5mm}
\beq
\barr{l}
\ds \bmA^{(c)} = \frac{1}{2(b^2 - d^2)(b + e)} \times \\[3mm]
\ds
\left\{
\frac{1}{2} \big[bd\gamma(b + e) - b\alpha(be + d^2) + (b^2 - d^2)f\big] \bmE_2^{(c)} -
\frac{1}{2} b(b\gamma - d\alpha)(b + e) \mQ^{(c)}\bE_3\left(\frac{\partial}{\partial x_1},\frac{\partial}{\partial x_3}\right)
\right\},
\earr
\eeq
\vspace{5mm}
\beq
\bmB^{(c)} = \frac{1}{2(b^2 - d^2)(b + e)}
\ds \left\{ b(b + e) \bI \mQ^{(c)} \left(\frac{\partial^2}{\partial x_1^2} + \frac{\partial^2}{\partial x_3^2}\right) -
(be + d^2) \mQ^{(c)} \bE_2\left(\frac{\partial}{\partial x_1},\frac{\partial}{\partial x_3}\right) \right\},
\eeq
and the matrix operator $\bmE_2^{(s,c)}$ is defined as
\beq
\bmE_2^{(s,c)} =
\left[
\left[
\barr{c}
\mQ_3^{(s,c)} \\[3mm]
0 \\[3mm]
-\mQ_1^{(s,c)}
\earr
\right]
\otimes
\left[
\barr{c}
\ds \frac{\partial}{\partial x_3} \\[5mm]
0 \\[3mm]
\ds -\frac{\partial}{\partial x_1}
\earr
\right]
\right]^T =
\left[
\barr{ccc}
\ds \mQ_{3}^{(s,c)} \frac{\partial}{\partial x_3} & 0 & \ds -\mQ_{1}^{(s,c)} \frac{\partial}{\partial x_3} \\[5mm]
0 & 0 & 0 \\[3mm]
\ds -\mQ_{3}^{(s,c)} \frac{\partial}{\partial x_1} & 0 & \ds \mQ_{1}^{(s,c)} \frac{\partial}{\partial x_1}
\earr
\right].
\eeq
In the case of homogeneous body the matrix operators $\bmA^{(s)}$, $\bmB^{(s)}$, $\bmA^{(c)}$, $\bmB^{(c)}$ simplify to
\beq
\label{om1}
\bmA^{(s,c)} = -\frac{1 - 2\nu}{8(1 - \nu)} \mQ^{(s,c)}
\left[
\barr{ccc}
0 & \ds -\frac{\partial}{\partial x_1} & 0 \\[3mm]
\ds \frac{\partial}{\partial x_1} & 0 & \ds \frac{\partial}{\partial x_3} \\[3mm]
0 & \ds -\frac{\partial}{\partial x_3} & 0
\earr
\right],
\eeq
\beq
\label{om2}
\bmB^{(s,c)} = \frac{\mu}{4(1 - \nu)} \mQ^{(s,c)}
\left[
\barr{ccc}
\ds \frac{\partial^2}{\partial x_1^2} + (1 - \nu)\frac{\partial^2}{\partial x_3^2} & 0 & \ds \nu\frac{\partial^2}{\partial x_1 \partial x_3} \\[3mm]
0 & \ds \frac{\partial^2}{\partial x_1^2} + \frac{\partial^2}{\partial x_3^2} & 0 \\[3mm]
\ds \nu\frac{\partial^2}{\partial x_1 \partial x_3} & 0 & \ds (1 - \nu)\frac{\partial^2}{\partial x_1^2} + \frac{\partial^2}{\partial x_3^2}
\earr
\right].
\eeq

We note that the singular integral formulation (\ref{intid3d1}) for the case of a homogeneous body and symmetrical forces ($\jump{0.15}{\bp} = 0$) is equivalent
to the formulation given by Weaver (1977) and Budiansky and Rice (1979). To show this it is sufficient to use integration by parts, taking into account the behaviour 
of the displacements jump at infinity.

\section{Conclusions}
The singular integral formulation for a three-dimensional semi-infinite interfacial crack has been derived in explicit form by means of fundamental properties of 
linear elasticity and integral transforms. This formulation for interfacial cracks seems to be unknown in the literature. The method of derivation avoids the use of 
Green's function and the tedious limiting process involved in the general procedure. The identities are important in applications especially in case of multiphysics 
problems where the elasticity equations are coupled with other concurrent phenomena, such as in hydraulic fracturing.

\vspace{10mm}
{\bf Acknowledgement.} A.P. and G.M. gratefully acknowledge the support from the European Union Seventh Framework Programme
under contract numbers PIEF-GA-2009-252857 and PIAP-GA-2009-251475, respectively.

\appendix
\renewcommand{\theequation}{\thesection.\arabic{equation}}

\section{Inversion of the 2D singular operator $S^{(s)}$}
\setcounter{equation}{0}
\label{appS}

The inversion formula for the singular integral equation
\beq
\label{int1}
\psi(x_1) = S^{(s)}\varphi(x_1) =
\frac{1}{\pi} \int_{-\infty}^{0} \frac{\varphi(\xi)}{x_1 - \xi} d\xi, \quad x_1 < 0,
\eeq
strongly depends on the properties of the known function $\psi(x_1)$. Indeed, let us first assume that the function $\psi(x_1)$ has compact support
$-a \le x_1 \le -b$, where $a$ and $b$ are positive constants, and belongs to a H\"older class, as it usually takes place in the classical 2D elasticity
(Muskhelishvili, 1946). Then the inversion formula can be found in Muskhelishvili (1946) (see also Rice, 1968)
\beq
\label{rice}
\varphi(x_1) = \left(S^{(s)}\right)_1^{-1}\psi(x_1) = -\frac{1}{\pi} \int_{-\infty}^{0} \sqrt{\frac{\xi}{x_1}} \frac{\psi(\xi)}{x_1 - \xi} d\xi, \quad x_1 < 0.
\eeq
Note that under such assumptions
\beq
\label{rice_0}
\varphi(x_1) \sim \frac{K_0}{\sqrt{-x_1}}, \quad x_1 \to 0^-, \quad
\varphi(x_1) \sim \frac{K_\infty}{(-x_1)^{3/2}} , \quad x_1 \to -\infty.
\eeq
Here the constants $K_0$ and $K_\infty$ ($K_0$ is the so-called stress intensity factor) are determined by the formulae
\beq
\label{rice_1}
K_0 = -\frac{1}{\pi} \int_{-\infty}^{0} \frac{\psi(\xi)}{\sqrt{-\xi}} d\xi, \quad
K_\infty = \frac{1}{\pi} \int_{-\infty}^{0} \psi(\xi)\sqrt{-\xi} d\xi.
\eeq

Let us now consider the case where the function $\psi(x_1)$ extends over the whole negative semi-axis, having the following behaviour at zero and infinity:
\beq
\label{inf}
\psi(x_1) \sim \frac{\psi_0}{(-x_1)^{\alpha_0}}, \quad x_1 \to 0^-, \quad
\psi(x_1) \sim \frac{\psi_\infty}{(-x_1)^{\alpha_\infty}}, \quad x_1 \to -\infty.
\eeq

If we assume that $\alpha_0<1/2$ and $\alpha_\infty>3/2$ then the inversion formula (\ref{rice}) is still valid and leads to the asymptotics (\ref{rice_0}).
In the case $\alpha_0<1/2$ and $1/2<\alpha_\infty<3/2$ the inversion formula (\ref{rice}) remains true, the behaviour of the function $\psi(x_1)$ near zero is the
same as in (\ref{rice_0})$_1$, while the behaviour at infinity (\ref{rice_0})$_2$ changes to:
\beq
\label{inf_1}
\varphi(x_1) \sim \frac{-\pi\psi_\infty\sin(\pi\alpha_\infty)}{(-x_1)^{\alpha_\infty}}, \quad x_1 \to -\infty.
\eeq

However, there are situations (such as some problems arising in hydraulic fracturing, see Garagash and Detournay, 2000) where the behaviour of the function $\psi(x_1)$
at infinity is worst, so that $\alpha_\infty$ in (\ref{inf}) is in the range $0 < \alpha_\infty < 1/2$. In such cases, the classic inversion formula (\ref{rice}) is not
any longer applicable and another inversion formula should be used instead. In Garagash and Detournay (2000) the formula
\beq
\label{garag}
\varphi(x_1) = \left(S^{(s)}\right)_2^{-1}\psi(x_1) =
-\frac{C_0}{(-x_1)^{1/2}} - \frac{1}{\pi} \int_{-\infty}^{0} \sqrt{\frac{x_1}{\xi}} \frac{\psi(\xi)}{x_1 - \xi} d\xi, \quad x_1 < 0,
\eeq
was derived by easy and elegant arguments leaving, however, the constant $C_0$ unknown.

Here we present an inversion formula accurately derived under slightly more restricted conditions on the behaviour of the function $\psi(x_1)$ at infinity. Namely we
assume that
\beq
\psi(x_1) = \psi_\infty(x_1) + O\left[\frac{1}{(-x_1)^{\beta}}\right], \quad x_1 \to -\infty,
\eeq
where
\beq
\psi_\infty(x_1) = \sum_{j = 1}^N \frac{\psi^{(j)}_\infty}{(-x_1)^{\alpha^{(j)}_\infty}},
\eeq
and $0 < \alpha^{(j)}_\infty < 1/2$ and $\beta > 1/2$. The alternative formula reads
\beq
\varphi(x_1) = \left(S^{(s)}\right)_2^{-1}\psi(x_1) =
\sum_{j = 1}^N \tan(\pi\alpha_\infty^{(j)}) \frac{\psi^{(j)}_\infty}{(-x_1)^{\alpha^{(j)}_\infty}} -
\frac{1}{\pi} \int_{-\infty}^{0} \sqrt{\frac{\xi}{x_1}} \frac{\psi(\xi) - \psi_\infty(\xi)}{x_1 - \xi} d\xi, \quad x_1 < 0.
\eeq
Note that this formula can be rewritten in the equivalent form (\ref{garag}), where the constant $C_0$ can be determined by
\beq
C_0 = \frac{1}{\pi} \int_{-\infty}^{0} \frac{\psi(\xi) - \psi_\infty(\xi)}{(-\xi)^{1/2}} d\xi.
\eeq

\section{Reduction of the identities from 3D to 2D case}
\setcounter{equation}{0}
\label{appredu}

The integral identities for 3D elasticity can be reduced to the 2D case by assuming that all the mechanical fields are independent of the $x_3$ coordinate.
We start by showing that if the fields $\bu$, $\bsigma_2$, $\bp$ are independent of $x_3$ then the convolution with respect to $x_1$ and $x_3$, denoted by
$\circledast$, reduces to the convolution with respect to the $x_1$ coordinate only, denoted by $*$.

In particular, we have the following formulae (proved in Appendix \ref{appid}, Theorem \ref{th1}). If $\varphi$ is a function of $x_1$ only, then
\beq
\label{form1}
\frac{1}{\sqrt{x_1^2 + x_3^2}} \circledast \frac{\partial^2 \varphi}{\partial x_1^2} =
-2 \frac{1}{x_1} * \frac{\partial \varphi}{\partial x_1}
\quad \rightsquigarrow \quad
\mQ \frac{\partial^2 \varphi}{\partial x_1^2} = -2 \mS \frac{\partial \varphi}{\partial x_1}
\eeq
\beq
\label{form2}
\frac{1}{\sqrt{x_1^2 + x_3^2}} \circledast \frac{\partial \varphi}{\partial x_1} =
-2 \frac{1}{x_1} * \varphi \quad \rightsquigarrow \quad
\mQ \frac{\partial \varphi}{\partial x_1} = -2 \mS \varphi.
\eeq
We also make use of the following identities (proved in Appendix \ref{appid}, Theorems \ref{th2} and \ref{th3}):
\beq
\label{iden1}
\frac{x_1}{x_1^2 + x_3^2} \circledast \frac{\partial \varphi}{\partial x_1} + \frac{x_3}{x_1^2 + x_3^2} \circledast \frac{\partial \varphi}{\partial x_3} =
2\pi \varphi
\quad \rightsquigarrow \quad
\mQ_1 \frac{\partial \varphi}{\partial x_1} + \mQ_3 \frac{\partial \varphi}{\partial x_3} = 2 \varphi,
\eeq
\beq
\label{iden2}
\frac{x_3}{x_1^2 + x_3^2} \circledast \frac{\partial \varphi}{\partial x_1} =
\frac{x_1}{x_1^2 + x_3^2} \circledast \frac{\partial \varphi}{\partial x_3}
\quad \rightsquigarrow \quad
\mQ_3 \frac{\partial \varphi}{\partial x_1} =
\mQ_1 \frac{\partial \varphi}{\partial x_3}.
\eeq
If $\varphi$ is a function of $x_1$ only then the identities (\ref{iden1}) and (\ref{iden2}) yield
\beq
\label{iden1b}
\frac{x_1}{x_1^2 + x_3^2} \circledast \frac{\partial \varphi}{\partial x_1} =
2\pi \varphi
\quad \rightsquigarrow \quad
\mQ_1 \frac{\partial \varphi}{\partial x_1} = 2 \varphi,
\eeq
\beq
\label{iden2b}
\frac{x_3}{x_1^2 + x_3^2} \circledast \frac{\partial \varphi}{\partial x_1} = 0
\quad \rightsquigarrow \quad
\mQ_3 \frac{\partial \varphi}{\partial x_1} = 0.
\eeq

Using the formulae (\ref{form1}), (\ref{form2}), (\ref{iden1b}) and (\ref{iden2b}), the integral identity (\ref{intid3d1}) for 3D elasticity reduces to
\beq
\langle \bp \rangle +
\frac{b}{2 (b^2 - d^2)}
\left[
\barr{ccc}
\ds (b\alpha - d\gamma) & \ds (d\alpha - b\gamma) \mS^{(s)} & 0 \\[3mm]
\ds -(d\alpha - b\gamma) \mS^{(s)} & \ds (b\alpha - d\gamma) & 0 \\[3mm]
0 & 0 & 0
\earr
\right]
\jump{0.15}{\bp} +
\frac{\eta}{2}
\left[
\barr{ccc}
0 & 0 & 0 \\[3mm]
0 & 0 & 0 \\[3mm]
0 & 0 & 1
\earr
\right]
\jump{0.15}{\bp} =
\eeq
\beq
-\frac{1}{b^2 - d^2}
\left[
\barr{ccc}
\ds b \mS^{(s)} \frac{\partial}{\partial x_1} & \ds -d \frac{\partial}{\partial x_1} & 0 \\[3mm]
\ds d \frac{\partial}{\partial x_1} & \ds \ds b \mS^{(s)} \frac{\partial}{\partial x_1} & 0 \\[3mm]
0 & 0 & 0
\earr
\right] \jump{0.15}{\bu}^{(-)}
-\frac{1}{b + e}
\left[
\barr{ccc}
0 & 0 & 0 \\[3mm]
0 & 0 & 0 \\[3mm]
0 & 0 & \ds \mS^{(s)} \frac{\partial}{\partial x_1}
\earr
\right] \jump{0.15}{\bu}^{(-)}.
\eeq
This formula is fully consistent with the formulation for 2D elasticity given in Secs. \ref{secmode3} and \ref{secmode12}. In a similar way, also the formula
(\ref{intid3d2}) can be shown to reduce to the 2D formulation.

\section{Illustrative example in 2D case: point forces applied at the crack faces}
\setcounter{equation}{0}
\label{appex}

In this section, we illustrate the use of the singular integral identities in the 2D case, when the loading is given in terms of point forces applied at the crack
faces. We analyse first the antiplane problem and then the vector problem.

{\bf Mode III. Symmetrical point forces.} Assume that the loading is given as two symmetrical point forces applied on the crack faces, at a distance $a$ from the
crack tip, and directed along the $x_3$-axis,
\beq
\langle p_3 \rangle(x_1) = -F \delta(x_1 + a), \quad \jump{0.15}{p_3}(x_1) = 0,
\eeq
where $\delta$ is the Dirac delta function.

The singular integral equation relating the applied loading with the resulting crack opening is given by (\ref{res1}) and, by means of the inversion formula (\ref{rice})
we obtain
\beq
\label{der}
\frac{\partial \jump{0.15}{u_3}^{(-)}}{\partial x_1} = -\frac{b + e}{\pi} F \int_{-\infty}^0 \sqrt{\frac{\xi}{x_1}} \frac{\delta(\xi + a)}{x_1 - \xi} d\xi =
- \frac{b + e}{\pi} F \sqrt{-\frac{a}{x_1}} \frac{1}{x_1 + a}.
\eeq
By integration of (\ref{der}) and using the conditions that the crack opening vanishes at zero and at infinity, we get
\beq
\jump{0.15}{u_3}(x_1) =
\left\{
\barr{l}
\ds \frac{2F}{\pi}(b + e) \arctanh \sqrt{-\frac{x_1}{a}}, \quad -a < x_1 < 0, \\[5mm]
\ds \frac{2F}{\pi}(b + e) \arctanh \sqrt{-\frac{a}{x_1}}, \quad x_1 < -a.
\earr
\right.
\eeq
From (\ref{res2}) we can now obtain the corresponding tractions ahead of the crack tip, namely
\beq
\langle \sigma_{23} \rangle^{(+)}(x_1) = - \frac{1}{\pi(b + e)} \int_{-\infty}^0 \frac{1}{x_1 - \xi} \frac{\partial \jump{0.15}{u_3}^{(-)}}{\partial \xi} d\xi =
\frac{F}{\pi} \sqrt{\frac{a}{x_1}} \frac{1}{x_1 + a}.
\eeq
Note that the stress intensity factor in this case is
\beq
K_\modIII = \lim_{x_1 \to 0} \sqrt{2\pi x_1} \langle \sigma_{23} \rangle^{(+)}(x_1) = \sqrt{\frac{2}{\pi a}} F.
\eeq

{\bf Mode III. Skew-symmetrical point forces.} We assume now that the loading is given by two skew-symmetrical point forces again applied at a distance $a$ behind
crack tip and directed along the $x_3$-axis,
\beq
\langle p_3 \rangle(x_1) = 0, \quad \jump{0.15}{p_3}(x_1) = -2F \delta(x_1 + a).
\eeq
It is easy to show that all the formulae derived above still apply, however, they contain now the factor $\eta$, so that
\beq
\frac{\partial \jump{0.15}{u_3}^{(-)}}{\partial x_1} = - \eta \frac{b + e}{\pi} F \sqrt{-\frac{a}{x_1}} \frac{1}{x_1 + a},
\eeq
\beq
\jump{0.15}{u_3}(x_1) =
\left\{
\barr{l}
\ds \eta \frac{2F}{\pi}(b + e) \arctanh \sqrt{-\frac{x_1}{a}}, \quad -a < x_1 < 0, \\[5mm]
\ds \eta \frac{2F}{\pi}(b + e) \arctanh \sqrt{-\frac{a}{x_1}}, \quad x_1 < -a,
\earr
\right.
\eeq
\beq
\langle \sigma_{23} \rangle^{(+)}(x_1) = \eta \frac{F}{\pi} \sqrt{\frac{a}{x_1}} \frac{1}{x_1 + a},
\eeq
and the stress intensity factor is given by
\beq
K_\modIII = \eta \sqrt{\frac{2}{\pi a}} F.
\eeq

{\bf Mode I and II. Symmetrical point forces.} We consider now the vector 2D case and suppose that the loading is given as symmetrical point forces in both $x_1$ and
$x_2$ directions, namely
\beq
\langle p_1 \rangle(x_1) = -F_1 \delta(x_1 + a), \quad \jump{0.15}{p_1}(x_1) = 0,
\eeq
\beq
\langle p_2 \rangle(x_1) = -F_2 \delta(x_1 + a), \quad \jump{0.15}{p_2}(x_1) = 0.
\eeq
For simplicity we assume that the Dundurs parameter $d$ is equal to zero, so that the system of singular equations (\ref{general_2D}) decouples
\beq
-\frac{1}{b} \mS^{(s)} \frac{\partial \jump{0.15}{u_j}^{(-)}}{\partial x_1} = \langle p_j \rangle(x_1), \quad j = 1,2.
\eeq
and it is possible to obtain results in closed form. Indeed, the inversion formula (\ref{rice}) leads to
\beq
\frac{\partial \jump{0.15}{u_j}^{(-)}}{\partial x_1} = - \frac{b F_j}{\pi} \sqrt{-\frac{a}{x_1}} \frac{1}{x_1 + a},
\eeq
which, after integration, gives
\beq
\jump{0.15}{u_j}(x_1) =
\left\{
\barr{l}
\ds \frac{2b F_j}{\pi} \arctanh \sqrt{-\frac{x_1}{a}}, \quad -a < x_1 < 0, \\[5mm]
\ds \frac{2b F_j}{\pi} \arctanh \sqrt{-\frac{a}{x_1}}, \quad x_1 < -a.
\earr
\right.
\eeq
Finally, the tractions ahead of the crack tip are calculated from (\ref{general_2D2})
\beq
\langle \sigma_{2j} \rangle^{(+)}(x_1) = \frac{F_j}{\pi} \sqrt{\frac{a}{x_1}} \frac{1}{x_1 + a}.
\eeq
The stress intensity factors are then obtained as
\beq
K_\modI = \sqrt{\frac{2}{\pi a}} F_2, \quad K_\modII = \sqrt{\frac{2}{\pi a}} F_1.
\eeq

{\bf Mode I and II. Skew-symmetrical point forces.} In this case the point forces, applied at a distance $a$ behind the crack tip, have the same direction
\beq
\langle p_1 \rangle(x_1) = 0, \quad \jump{0.15}{p_1}(x_1) = -2F_1 \delta(x_1 + a),
\eeq
\beq
\langle p_2 \rangle(x_1) = 0, \quad \jump{0.15}{p_2}(x_1) = -2F_2 \delta(x_1 + a).
\eeq
Assuming again $d = 0$, the solution of the system of singular integral equations (\ref{general_2D}) leads to
\beq
\frac{\partial \jump{0.15}{u_1}^{(-)}}{\partial x_1} = - \frac{b\alpha F_1}{\pi} \sqrt{-\frac{a}{x_1}} \frac{1}{x_1 + a} - b\gamma F_2 \delta(x_1 + a),
\eeq
\beq
\frac{\partial \jump{0.15}{u_2}^{(-)}}{\partial x_1} = - \frac{b\alpha F_2}{\pi} \sqrt{-\frac{a}{x_1}} \frac{1}{x_1 + a} + b\gamma F_1 \delta(x_1 + a),
\eeq
which after integration gives
\beq
\label{open}
\barr{l}
\ds
\jump{0.15}{u_1}(x_1) =
\left\{
\barr{ll}
\ds \frac{2b\alpha F_1}{\pi} \arctanh \sqrt{-\frac{x_1}{a}}, & -a < x_1 < 0, \\[5mm]
\ds \frac{2b\alpha F_1}{\pi} \arctanh \sqrt{-\frac{a}{x_1}} + b\gamma F_2, & x_1 < -a,
\earr
\right.
\\[12mm]
\ds
\jump{0.15}{u_2}(x_1) =
\left\{
\barr{ll}
\ds \frac{2b\alpha F_2}{\pi} \arctanh \sqrt{-\frac{x_1}{a}}, & -a < x_1 < 0, \\[5mm]
\ds \frac{2b\alpha F_2}{\pi} \arctanh \sqrt{-\frac{a}{x_1}} - b\gamma F_1, & x_1 < -a.
\earr
\right.
\earr
\eeq
The tractions ahead of the crack tip are obtained from (\ref{general_2D2}) and read
\beq
\langle \sigma_{21} \rangle^{(+)}(x_1) =
-\frac{1}{b} \mS^{(c)} \frac{\partial \jump{0.15}{u_1}^{(-)}}{\partial x_1} + \frac{\gamma}{2} \mS^{(c)} \jump{0.15}{p_2}(x_1) =
\alpha \frac{F_1}{\pi} \sqrt{\frac{a}{x_1}} \frac{1}{x_1 + a},
\eeq
\beq
\langle \sigma_{22} \rangle^{(+)}(x_1) =
-\frac{1}{b} \mS^{(c)} \frac{\partial \jump{0.15}{u_2}^{(-)}}{\partial x_1} - \frac{\gamma}{2} \mS^{(c)} \jump{0.15}{p_1}(x_1) =
\alpha \frac{F_2}{\pi} \sqrt{\frac{a}{x_1}} \frac{1}{x_1 + a}.
\eeq
The stress intensity factors are then obtained as
\beq
K_\modI = \alpha \sqrt{\frac{2}{\pi a}} F_2, \quad K_\modII = \alpha \sqrt{\frac{2}{\pi a}} F_1.
\eeq
The reader may be surprised by the coupling constant terms in the right-hand side of expressions (\ref{open}), showing a discontinuity of the crack opening
at the point of application of the forces. However, these terms can be easily explained by making recourse to the Flamant solution for a half-plane loaded by 
concentrated forces $F_1$ and $F_2$ at its surface (Barber, 2002). 

Assuming a reference system centred at the point of application of the forces, the surface displacement are determined, up to
an arbitrary constant, by 
\beq
\label{flam}
\barr{l}
\ds u_1(x_1) = -\frac{F_1 (1 - \nu) \log|x_1|}{\pi \mu} + \frac{F_2(1 - 2\nu) \sign(x_1)}{4\mu}, \\[6mm]
\ds u_2(x_1) = -\frac{F_1(1 - 2\nu) \sign(x_1)}{4\mu} - \frac{F_1 (1 - \nu) \log|x_1|}{\pi \mu},
\earr
\eeq
where $\mu$ and $\nu$ are the shear modulus and the Poisson's ratio of the half-plane, respectively.

We now apply this solution to the two half-planes making up the cracked bimaterial
plane. Of course, we do not expect this solution to match exactly with the correct solution obtained above, since the boundary conditions for the crack problem differ
from the Flamant problem along the interface joining the two materials. Nonetheless, let us for the time being make this analogy and calculate from the solution
(\ref{flam}) what the crack opening would be.

We consider only the coupling (diagonal) terms in the system (\ref{flam}), which can explain the observed discontinuity in the crack opening. We also adjust the arbitrary
constant in the Flamant solution in order to have the same displacement for the two crack faces at the crack tip.
As a result, we obtain
\beq
\jump{0.15}{u_1}(x_1) = F_2 \sign(-x_1 - a) \left( \frac{1 - 2\nu_+}{2\mu_+} + \frac{1 - 2\nu_-}{2\mu_-} \right) = b\gamma F_2 \sign(-x_1 - a),
\eeq
\beq
\jump{0.15}{u_1}(x_1) = -F_1 \sign(-x_1 - a) \left( \frac{1 - 2\nu_+}{2\mu_+} + \frac{1 - 2\nu_-}{2\mu_-} \right) = -b\gamma F_1 \sign(-x_1 - a),
\eeq
which is in perfect agreement with our solution (\ref{open}).

\section{Identities involving convolution integrals}
\setcounter{equation}{0}
\label{appid}

\begin{theorem}
\label{th1}
Let $\varphi(x_1)$ be a function defined on $\Reals$ and possess the derivatives of order $1,...,n$. Assume also that the function $\varphi(x_1)$ is
vanishing with all its derivatives as $x_1 \to \pm \infty$. Then
\beq
\int_{-\infty}^{\infty} \int_{-\infty}^{\infty} \frac{1}{\sqrt{(x_1 - \xi_1)^2 + (x_3 - \xi_3)^2}} \frac{\partial^n \varphi}{\partial \xi_1^n} d\xi_1 d\xi_3 =
-2 \int_{-\infty}^{\infty} \frac{1}{x_1 - \xi_1} \frac{\partial^{n-1} \varphi}{\partial x_1^{n-1}} d\xi_1.
\eeq
\end{theorem}

\begin{proof}
Consider the integral
\beq
\int_{-R}^{R} \int_{-\infty}^{\infty} \frac{1}{\sqrt{(x_1 - \xi_1)^2 + (x_3 - \xi_3)^2}} \frac{\partial^n \varphi}{\partial \xi_1^n} d\xi_1 d\xi_3,
\eeq
and integrate by parts the inner integral, to obtain
\beq
\label{part}
\int_{-R}^{R} \left. \frac{1}{\sqrt{(x_1 - \xi_1)^2 + (x_3 - \xi_3)^2}}
\frac{\partial^{n-1} \varphi}{\partial \xi_1^{n-1}} \right|_{\xi_1 = -\infty}^{\xi_1 = \infty} d\xi_3 -
\int_{-R}^{R} \int_{-\infty}^{\infty} \frac{x_1 - \xi_1}{[(x_1 - \xi_1)^2 + (x_3 - \xi_3)^2]^{3/2}}
\frac{\partial^{n-1} \varphi}{\partial \xi_1^{n-1}} d\xi_1 d\xi_3.
\eeq
The first term in (\ref{part}) is zero and the second term can be rewritten changing the order of integration in the form
\beq
\label{inner}
-\int_{-\infty}^{\infty} \left( \int_{-R}^{R} \frac{1}{[(x_1 - \xi_1)^2 + (x_3 - \xi_3)^2]^{3/2}} d\xi_3 \right) (x_1 - \xi_1)
\frac{\partial^{n-1} \varphi}{\partial \xi_1^{n-1}} d\xi_1.
\eeq
Taking the limit as $R \to \infty$, the inner integral in (\ref{inner}) gives
\beq
\lim_{R \to \infty} \int_{-R}^{R} \frac{1}{[(x_1 - \xi_1)^2 + (x_3 - \xi_3)^2]^{3/2}} d\xi_3 = \frac{2}{(x_1 - \xi_1)^2},
\eeq
which concludes the proof.
\end{proof}

\begin{theorem}
\label{th2}
Let $\varphi(x_1,x_3)$ be a function defined on $\Reals^2$ and possess the Fourier transform with respect to $x_1$ and $x_3$, denoted by
$\overline{\varphi}(\beta,\lambda)$. Then
\beq
\barr{l}
\ds \int_{-\infty}^{\infty} \int_{-\infty}^{\infty} \frac{x_1 - \xi_1}{(x_1 - \xi_1)^2 + (x_3 - \xi_3)^2} \frac{\partial \varphi}{\partial \xi_1} d\xi_1 d\xi_3 +
\int_{-\infty}^{\infty} \int_{-\infty}^{\infty} \frac{x_3 - \xi_3}{(x_1 - \xi_1)^2 + (x_3 - \xi_3)^2} \frac{\partial \varphi}{\partial \xi_3} d\xi_1 d\xi_3 = \\[6mm]
\ds \hspace{120mm} = 2\pi \varphi(x_1,x_3).
\earr
\eeq
\end{theorem}

\begin{proof}
Consider the Fourier transform $\overline{\varphi}(\beta,\lambda)$ and write
\beq
\label{theo2}
\overline{\varphi}(\beta,\lambda) = \frac{\beta^2}{\beta^2 + \lambda^2} \cdot \overline{\varphi}(\beta,\lambda) +
\frac{\lambda^2}{\beta^2 + \lambda^2} \cdot \overline{\varphi}(\beta,\lambda).
\eeq
To conclude the proof, apply the inversion Fourier transform to both side of (\ref{theo2}) and use the first two formulae given in Tab. \ref{tab01}.
\end{proof}

\begin{theorem}
\label{th3}
Let $\varphi(x_1,x_3)$ be a function defined on $\Reals^2$ and possess the Fourier transform with respect $x_1$ and $x_3$, denoted by
$\overline{\varphi}(\beta,\lambda)$. Then
\beq
\int_{-\infty}^{\infty} \int_{-\infty}^{\infty} \frac{x_3 - \xi_3}{(x_1 - \xi_1)^2 + (x_3 - \xi_3)^2} \frac{\partial \varphi}{\partial \xi_1} d\xi_1 d\xi_3 =
\int_{-\infty}^{\infty} \int_{-\infty}^{\infty} \frac{x_1 - \xi_1}{(x_1 - \xi_1)^2 + (x_3 - \xi_3)^2} \frac{\partial \varphi}{\partial \xi_3} d\xi_1 d\xi_3.
\eeq
\end{theorem}

\begin{proof}
Consider the Fourier transform $\overline{\varphi}(\beta,\lambda)$ and write
\beq
\label{theo3}
\frac{\beta}{\beta^2 + \lambda^2} \cdot \lambda \overline{\varphi}(\beta,\lambda) = \frac{\lambda}{\beta^2 + \lambda^2} \cdot \beta \overline{\varphi}(\beta,\lambda).
\eeq
To conclude the proof, apply the inversion Fourier transform to both side of (\ref{theo2}) and use the third and fourth formulae in Tab. \ref{tab01}.
\end{proof}

\section{Double Fourier transform inversion for the 3D case}
\setcounter{equation}{0}
\label{appdoub}

The following four integrals will be used in the analysis (Gradshteyn and Ryzhik, 1965):
\beq
\label{1}
\int_0^\infty \frac{\cos ax}{\sqrt{\beta^2 + x^2}} dx = K_0(a\beta),\quad
a > 0,\quad \Re \beta > 0,
\eeq
\beq
\label{2}
\int_0^\infty \frac{x\sin ax}{\beta^2 + x^2} dx = \frac{\pi}{2} e^{-a\beta}, \quad a > 0, \quad \Re \beta > 0,
\eeq
\beq
\label{1a}
\int_0^\infty K_0(a\beta) \cos(b\beta) d\beta = \frac{\pi}{2} \frac{1}{\sqrt{a^2 + b^2}}, \quad \Re a > 0,
\eeq
\beq
\label{2a}
\int_0^\infty e^{-a\beta} \cos(b\beta) d\beta = \frac{a}{a^2 + b^2}, \quad a > 0,
\end{equation}
where $K_0(z)$ is the modified Bessel function of the second kind of zero order, satisfying the differential equation
\beq
z^2 y'' + z y' - z^2 y = 0.
\eeq
In some sense, the last two integrals are the Fourier inversion formulae for the first two.

We also shall make use of the fact that the Fourier transform of the convolution is given by the product of the Fourier transforms, which leads to
\beq
\label{3}
\mF^{-1}_{\beta}\mF^{-1}_\lambda \left[ \bar g(\beta,\lambda) \cdot \bar f(\beta,\lambda) \right] =
\int_{-\infty}^\infty \int_{-\infty}^\infty g(x_1 - \xi_1,x_3 - \xi_3) f(\xi_1,\xi_3) d\xi_1 d\xi_3.
\eeq

It is evident what to do with the factors $\beta$ and $\lambda$. Namely, these factors arise when taking the Fourier transform of the derivatives instead of
the function itself:
\beq
\beta \bar f(\beta,\lambda) = i \overline{\frac{\partial f}{\partial x_1}}(\beta,\lambda), \quad
\lambda \bar f(\beta,\lambda) = i \overline{\frac{\partial f}{\partial x_3}}(\beta,\lambda).
\eeq
Analogous identities can be written for the factors $\beta^2$, $\lambda^2$ and $\beta\lambda$ involving higher order derivatives:
\beq
\beta^2 \bar f(\beta,\lambda) = -\overline{\frac{\partial^2 f}{\partial x_1^2}}(\beta,\lambda), \quad
\lambda^2 \bar f(\beta,\lambda) = -\overline{\frac{\partial^2 f}{\partial x_3^2}}(\beta,\lambda), \quad
\beta\lambda \bar f(\beta,\lambda) = -\overline{\frac{\partial^2 f}{\partial x_1 \partial x_3}}(\beta,\lambda).
\eeq
Taking into account the following identity:
\beq
\sqrt{\beta^2 + \lambda^2} = \frac{\beta^2 + \lambda^2}{\sqrt{\beta^2 + \lambda^2}},
\eeq
we need to evaluate only the following three {\it basic} integrals:
\beq
\label{uno}
\mF^{-1}_{\beta}\mF^{-1}_\lambda \left[\frac{1}{\sqrt{\beta^2 + \lambda^2}}\right] =
\frac{1}{4\pi^2} \int_{-\infty}^{\infty} \int_{-\infty}^{\infty} \frac{1}{\sqrt{\beta^2 + \lambda^2}} e^{-ix_1\beta} e^{-ix_3\lambda} d\beta d\lambda,
\eeq
\beq
\label{due}
\mF^{-1}_{\beta}\mF^{-1}_\lambda \left[\frac{\lambda}{\beta^2 + \lambda^2}\right] =
\frac{1}{4\pi^2} \int_{-\infty}^{\infty} \int_{-\infty}^{\infty} \frac{\lambda}{\beta^2 + \lambda^2} e^{-ix_1\beta} e^{-ix_3\lambda} d\beta d\lambda,
\eeq
\beq
\label{tre}
\mF^{-1}_{\beta}\mF^{-1}_\lambda \left[\frac{\beta}{\beta^2 + \lambda^2}\right] =
\frac{1}{4\pi^2} \int_{-\infty}^{\infty} \int_{-\infty}^{\infty} \frac{\beta}{\beta^2 + \lambda^2} e^{-ix_1\beta} e^{-ix_3\lambda} d\beta d\lambda.
\eeq
Let us consider the first integral (\ref{uno}). By (\ref{1}) we obtain
\beq
\label{uno.1}
\int_{-\infty}^{\infty} e^{-ix_3\lambda} \frac{d\lambda}{\sqrt{\beta^2 + \lambda^2}} =
2 \int_0^{\infty} \frac{\cos(|x_3|\lambda)}{\sqrt{\beta^2 + \lambda^2}} d\lambda =
2 \left\{
\barr{ll}
K_0(|x_3|\beta), & \Re\beta > 0, \\[3mm]
K_0(-|x_3|\beta), & -\Re\beta > 0.
\earr
\right.
\eeq
Making recourse to (\ref{1a}), we conclude that
\beq
\label{uno.2}
\barr{ll}
\ds \mF^{-1}_{\beta} \mF^{-1}_\lambda \left[\frac{1}{\sqrt{\beta^2 + \lambda^2}}\right] &
\ds = \frac{1}{2\pi^2} \int_{-\infty}^{\infty} e^{-ix_1\beta} K_0(|x_3| |\beta|) d\beta, \\[5mm]
 &
\ds = \frac{1}{\pi^2} \int_{0}^{\infty} K_0(|x_3|\beta)\cos(x_1\beta) d\beta, \\[5mm]
 &
\ds = \frac{1}{2\pi} \frac{1}{\sqrt{x_1^2 + x_3^2}}.
\earr
\end{equation}
We note the the function transforms into itself (multiplied by a constant factor) after double Fourier transform.

Let us now consider the second integral (\ref{due}). By means of (\ref{2}) we write
\beq
\label{due.1}
\int_{-\infty}^{\infty} e^{-ix_3\lambda} \frac{\lambda d\lambda}{\beta^2 + \lambda^2} =
-2i\sign(x_3) \int_0^\infty \frac{\lambda\sin(|x_3|\lambda)}{\beta^2 + \lambda^2} d\lambda =
-\pi i\sign(x_3)
\left\{
\barr{ll}
e^{-|x_3|\beta}, & \Re \beta > 0, \\[3mm]
e^{|x_3|\beta}, & -\Re \beta > 0.
\earr
\right.
\eeq
We can now use the formula (\ref{2a}) to obtain
\beq
\label{due.2}
\barr{ll}
\ds \mF^{-1}_{\beta}\mF^{-1}_\lambda \left[\frac{\lambda}{\beta^2 + \lambda^2}\right] &
\ds = \frac{-i\sign(x_3)}{4\pi} \int_{-\infty}^\infty e^{-ix_1\beta} e^{-|x_3||\beta|} d\beta, \\[5mm]
 &
\ds = \frac{-i\sign(x_3)}{2\pi} \int_0^\infty e^{-|x_3|\beta} \cos(x_1\beta) d\beta, \\[5mm]
 &
\ds = \frac{-i\sign(x_3)}{2\pi}\frac{|x_3|}{x_1^2 + x_3^2} = \frac{-i}{2\pi} \frac{x_3}{x_1^2 + x_3^2}.
\earr
\eeq
Again, the function transforms into itself (multiplied by a constant factor) after double Fourier transform.

The last integral (\ref{tre}) can be evaluated in a similar way and we have
\beq
\label{tre.2}
\mF^{-1}_{\beta}\mF^{-1}_\lambda \left[\frac{\beta}{\beta^2 + \lambda^2}\right] = \frac{-i}{2\pi} \frac{x_1}{x_1^2 + x_3^2}.
\eeq
Equations (\ref{uno.2}), (\ref{due.2}) and (\ref{tre.2}) are the basic integrals we use to perform the double Fourier inversion for the 3D case
in Sec.\ \ref{sec3d}. The complete list of inversion formulae is given in Tab.\ \ref{tab01}.

\begin{table}[!htcb]
\centering
\ra{2.2}
\begin{tabular}{@{}ll@{}}
\toprule
Function: &
Double Fourier inversion: \\
$\bar f(\beta,\lambda)$ &
$\mF_\beta^{-1}\mF_\lambda^{-1} [\bar f(\beta,\lambda)]$ \\
\midrule
$\ds \frac{\beta}{\beta^2 + \lambda^2} \cdot \beta \bar \varphi$ &
$\ds \frac{1}{2\pi} \frac{x_1}{x_1^2 + x_3^2} \circledast \frac{\partial \varphi}{\partial x_1}$ \\
$\ds \frac{\lambda}{\beta^2 + \lambda^2} \cdot \lambda \bar \varphi$ &
$\ds \frac{1}{2\pi} \frac{x_3}{x_1^2 + x_3^2} \circledast \frac{\partial \varphi}{\partial x_3}$ \\
$\ds \frac{\beta}{\beta^2 + \lambda^2} \cdot \lambda \bar \varphi$ &
$\ds \frac{1}{2\pi} \frac{x_1}{x_1^2 + x_3^2} \circledast \frac{\partial \varphi}{\partial x_3}$ \\
$\ds \frac{\lambda}{\beta^2 + \lambda^2} \cdot \beta \bar \varphi$ &
$\ds \frac{1}{2\pi} \frac{x_3}{x_1^2 + x_3^2} \circledast \frac{\partial \varphi}{\partial x_1}$ \\
$\ds \frac{1}{\sqrt{\beta^2 + \lambda^2}} \cdot \beta \bar \varphi$ &
$\ds \frac{i}{2\pi} \frac{1}{\sqrt{x_1^2 + x_3^2}} \circledast \frac{\partial \varphi}{\partial x_1}$ \\
$\ds \frac{1}{\sqrt{\beta^2 + \lambda^2}} \cdot \lambda \bar \varphi$ &
$\ds \frac{i}{2\pi} \frac{1}{\sqrt{x_1^2 + x_3^2}} \circledast \frac{\partial \varphi}{\partial x_3}$ \\
$\ds \frac{1}{\sqrt{\beta^2 + \lambda^2}} \cdot \beta^2 \bar \varphi$ &
$\ds -\frac{1}{2\pi} \frac{1}{\sqrt{x_1^2 + x_3^2}} \circledast \frac{\partial^2 \varphi}{\partial x_1^2}$ \\
$\ds \frac{1}{\sqrt{\beta^2 + \lambda^2}} \cdot \lambda^2 \bar \varphi$ &
$\ds -\frac{1}{2\pi} \frac{1}{\sqrt{x_1^2 + x_3^2}} \circledast \frac{\partial^2 \varphi}{\partial x_3^2}$ \\
$\ds \frac{1}{\sqrt{\beta^2 + \lambda^2}} \cdot \beta\lambda \bar \varphi$ &
$\ds -\frac{1}{2\pi} \frac{1}{\sqrt{x_1^2 + x_3^2}} \circledast \frac{\partial^2 \varphi}{\partial x_1 \partial x_3}$ \\
$\ds \sqrt{\beta^2 + \lambda^2} \cdot \bar \varphi$ &
$\ds -\frac{1}{2\pi} \frac{1}{\sqrt{x_1^2 + x_3^2}} \circledast
\left(\frac{\partial^2 \varphi}{\partial x_1^2} + \frac{\partial^2 \varphi}{\partial x_3^2}\right)$ \\
\bottomrule
\end{tabular}
\caption{Double Fourier inversion formulae}
\label{tab01}
\end{table}

\end{document}